\documentclass[letterpaper, 10 pt, conference]{ieeeconf}

\IEEEoverridecommandlockouts   

\usepackage{cite}
\usepackage{amsmath,amssymb,amsfonts}
\usepackage{algorithmic}
\usepackage{subfig}
\usepackage{graphicx}
\usepackage{comment}
\usepackage{textcomp}
\usepackage{mathtools}
\def\BibTeX{{\rm B\kern-.05em{\sc i\kern-.025em b}\kern-.08em
    T\kern-.1667em\lower.7ex\hbox{E}\kern-.125emX}}
\usepackage{calrsfs}

\usepackage[colorlinks=true]{hyperref}
\usepackage{algorithm}

\newtheorem{theorem}{Theorem}[section]

\newtheorem{definition}[theorem]{Definition}

\newtheorem{problem}[theorem]{Problem}

\newtheorem{remark}[theorem]{Remark}

\newtheorem{assumptions}[theorem]{Assumptions}

\DeclareMathAlphabet{\pazocal}{OMS}{zplm}{m}{n}

\newcommand{\mb}{\mathbb}

\renewcommand{\theequation}{\thesection.\arabic{equation}}
\renewcommand{\thefigure}{\thesection.\arabic{figure}}
\renewcommand{\thetable}{\thesection.\arabic{table}}

\newcommand{\sr}{\stackrel}

\newcommand{\rar}{\rightarrow}

\newcommand{\tri}{\sr{\triangle}{=}}

\newcommand{\bea}{\begin{eqnarray}}
\newcommand{\eea}{\end{eqnarray}}
\newcommand{\bes}{\begin{eqnarray*}}
\newcommand{\ees}{\end{eqnarray*}}
\newcommand{\bce}{\begin{center}}
\newcommand{\ece}{\end{center}}

\def\VR{\kern-\arraycolsep\strut\vrule &\kern-\arraycolsep}
\def\vr{\kern-\arraycolsep & \kern-\arraycolsep}

\newcommand{\ben}{\begin{enumerate}}
\newcommand{\een}{\end{enumerate}}

\newcommand{\bi}{\begin{itemize}}
\newcommand{\ei}{\end{itemize}}

\newcommand{\bp}{\begin{problem}}
\newcommand{\ep}{\end{problem}}
\newcommand{\hso}{\hspace{.1in}}
\newcommand{\hst}{\hspace{.2in}}

\newcommand{\bc}{\begin{center}}
\newcommand{\ec}{\end{center}}

\begin{document}

\title{\LARGE \bf Information-Theoretic Meta Dynamic Programming for Signalling and Control of POMDPs}

%% Bambos suggestion %%
%Meta Dynamic Programming  for Signalling and  Control of  POMDPs}

\author{Charalambos D. Charalambous, Stelios Louka, and Photios A. Stavrou 
%\thanks{``\IT{This work is funded by ... }'' }
\thanks{C. D. Charalambous and S. Louka are with the Department of Electrical and Computer Engineering, University of Cyprus, Nicosia 1678, Cyprus. {\tt\small email: \{chadcha,\ louka.stelios\}@ucy.ac.cy}  }
\thanks{P. A. Stavrou is with the Foundations and Algorithms Group, Communication Systems Department, EURECOM, France. {\tt\small email: fotios.stavrou@eurecom.fr}}
}

\maketitle

%\vspace*{-1.cm}

\begin{abstract}
In this paper, we study the information-theoretic characterization of simultaneous signalling and control over channels modeled by partially observable Markov decision processes (POMDPs). The problem is formulated as an optimization over randomized control strategies that maximize the directed information from actions to observations, subject to an average-cost constraint.  We derive a novel dynamic programming framework in which the state is defined on the space of conditional probability distributions, leading to a high-level ``meta'' dynamic program. Specifically, we show that two coupled information states, namely, the posterior distribution of the system state and a distribution over such posteriors, satisfy Markov recursions and provide sufficient statistics for optimal control. This structure enables the decomposition of optimal strategies into separated randomized policies that depend only on these information states. Our results establish necessary and sufficient conditions for optimality and unify classical stochastic control and information-theoretic formulations. In particular, we show that in the absence of signalling, the proposed framework reduces to the standard dynamic programming equations for POMDPs. The developed approach provides a principled foundation for analyzing and designing control systems with intrinsic information constraints.

\end{abstract}

\section{Introduction}
\label{sect:intro}

Stochastic optimal control and Markov decision theory is traditionally focused on optimizing performance criteria of dynamical systems under uncertainty, where control strategies are designed based on available information \cite{striebel1965,gihman-skorohod1979,kumar-varayia1986,caines1988,hernandezlerma-lasserre1996,petersen-james-dupuis:2000}. In partially observable settings, such as partially observable Markov decision processes (POMDPs), the controller does not have direct access to the system state and instead relies on noisy observations. A fundamental tool in this context is the notion of the information state (or belief state), which evolves according to a Markov recursion and serves as a sufficient statistic for optimal control \cite{striebel1965,kumar-varayia1986,charalambous-elliott1997}.

In classical stochastic control with symmetric (classical) information structures, it is well known that randomized control strategies do not improve performance compared to deterministic ones \cite{gihman-skorohod1979,hernandezlerma-lasserre1996}. In contrast, information theory demonstrates that randomization is essential for reliable communication over noisy channels, where encoding strategies must be stochastic to achieve capacity \cite{shannon1948,gallager1968,cover-thomas2006,ihara1993,han2003}. This contrast highlights a fundamental gap between control and communication: while control seeks to regulate system behavior, communication aims to transmit information reliably.

Motivated by emerging applications in networked and cyber-physical systems, where control and communication are inherently intertwined, this paper studies the problem of simultaneous signalling and control when the noisy channels are modeled by POMDPs. In such systems, control actions play a dual role: they influence the system dynamics and convey information which is encoded into the strategies. This perspective is closely related to recent developments in control under communication constraints and information theory for channels with feedback \cite{kourtellaris-charalambousIT2015_Part_1,charalambous-kourtellaris-loykaIT2015_Part_2,charalambous-kourtellaris-tzortzis:SICON2020,charalambous-kourtellaris-tzortzis:IEEEAC2017,charalambous-kourtellaris-tziortzis:SICON-2024,chen-berger2005,tatikonda-mitter2009,permuter-cuff-roy-weissman2010,permuter-asnani-weissman2013,stavrou-charalambous-kourtellarisIT2016}. It naturally leads to an information-theoretic formulation of control, in which the objective is not only to achieve optimal  control performance but also to maximize information transfer.

To formalize this trade-off, we adopt the concept of control-coding (CC) capacity, introduced in prior works on stochastic control systems and channels with feedback \cite{kourtellaris-charalambousIT2015_Part_1,charalambous-kourtellaris-loykaIT2015_Part_2,charalambous-kourtellaris-tzortzis:SICON2020,charalambous-kourtellaris-tzortzis:IEEEAC2017,charalambous-kourtellaris-tziortzis:SICON-2024,charalambous-louka:2025BC}. The CC capacity represents the maximum achievable signalling rate subject to an average cost constraint. This formulation embeds two competing objectives: (i) controlling the system using partial observations while satisfying a performance constraint, and (ii) maximizing the information conveyed through the system dynamics and observations, measured via directed information.

\subsection{Main Contributions}

The main contributions of this paper are as follows:
\begin{itemize}
\item We provide an equivalent information-theoretic formulation of the operational CC capacity for channels modeled by POMDPs as an optimization problem over randomized control strategies. The objective is expressed in terms of directed information \cite{stavrou-charalambous-kourtellarisIT2016}, subject to an average cost constraint (see Theorem 2.2). This establishes a rigorous bridge between stochastic control and information theory.
\item We show that the problem admits two key information states:
\begin{enumerate}
\item[(i)] the classical posterior distribution of the system state given past actions and observations (as in nonlinear filtering \cite{striebel1965,charalambous-elliott1997});
\item[(ii)] a higher-order distribution defined on the space of posterior distributions.
\end{enumerate}
We prove that both evolve according to Markov recursions and constitute sufficient statistics for optimal signalling and control strategies (see Theorems 3.2 and 3.3).
\item Using the above information states, we derive a novel higher level dynamic programming framework on the space of probability measures, extending classical dynamic programming for POMDPs \cite{kumar-varayia1986,petersen-james-dupuis:2000}. This “meta” dynamic program characterizes the optimal strategies and provides a structured way to analyze the underlying optimization problem (see Theorem 3.6).
\item We establish that optimal randomized strategies admit a separation structure, depending only on the identified information states. We further derive necessary and sufficient conditions for optimality via the proposed dynamic programming equations (see Theorems 3.6 and 3.7).
\item We show that, in the absence of signalling, the proposed framework reduces to the classical dynamic programming formulation for POMDPs [1], [3]. This demonstrates that our results recover standard stochastic control as a special case (see discussion following Theorem 2.2 and Section III).
\end{itemize}

The remainder of the paper is organized as follows. In Section II, we formulate the signalling and control problem for channels modeled as POMDPs, and introduce the notion of CC capacity. In Section III, we develop the meta dynamic programming framework and characterize the associated information states and optimal strategies. Finally, conclusions and future research directions are provided in Section IV.

\begin{figure*}
\centering
\includegraphics[width=0.9\textwidth]{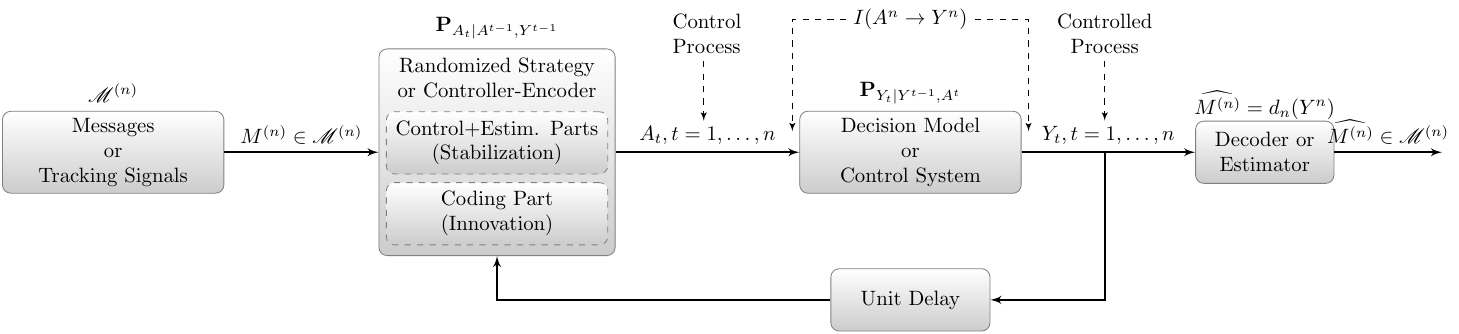}
\caption[]{\emph{Signalling and Control  Block Diagram: The equivalent channel is ${\bf P}_{Y_t|Y^{t-1}, A^i}, t=1, \ldots, n$.}}\label{fig_2}
\end{figure*}

%
% \begin{figure}
%\center
%\includegraphics[scale=0.6]{Figure_Encoder-Decoder.pdf}
%\label{fig1:dual}
%\caption{Communication block diagram and its analogy to stochastic optimal control.}
%\end{figure}  

\section{POMDPs, Signalling, and Control-Coding Capacity}\label{sec:problem.formulation}

In this section, we introduce the POMDP, the signalling and control protocol, and the characterization of control-coding (CC) capacity. We first present the system model and operational definition, then define CC capacity, and finally provide its equivalent information-theoretic formulation and its relation to classical stochastic control.

\subsection{Notation}

$\mathbb{R} \triangleq (-\infty,\infty)$, $\mathbb{Z}_+ \tri \{1,2, \ldots\}$, {$\mathbb{Z}_+^n \triangleq \{1,2, \ldots,n\}$,}  $n \in {\mathbb Z}_+$. $\{({\mathbb X}_t,{\cal B}({\mathbb X }_t))\big| t\in\mathbb{Z}_+^n\}$ denotes measurable spaces, where ${\mathbb X}_t$ is confined to complete separable metric space or Polish space, and ${\cal B}({\mathbb X}_t)$ is the Borel $\sigma-$algebra of subsets of ${\mathbb X}_t,  \forall t\in\mathbb{Z}_+^n$. Points in the product space ${\mathbb X}^n\triangleq{{\prod}_{t\in\mathbb{Z}_+^n}}{\mathbb X}_t$ are denoted by $x^n\triangleq \{x_1,\ldots, x_n\} \in{\mathbb X}^n$. $\big(\Omega, {\cal F}, {\mathbb P}\big)$   denotes a probability space.  Given a random variable (RV)  $X: (\Omega, {\cal F}) \rar  ({\mathbb X}, {\cal B}({\mathbb X}))$,  we denote  by\footnote{The subscript is omitted, if it is clear from the arguments of the measure.} ${\mathbb P}\big\{X \in dx\big\}={\bf P}_{X}(dx)\equiv {\bf P}(dx)$ the probability measure (PM) induced by $X$ on $({\mathbb X}, {\cal B}({\mathbb X}))$. ${\cal M}({\mathbb X})$ denotes the set of all probability measures on ${\mathbb X}$. 
\par Given another  RV, $Y: (\Omega, {\cal F}) \rar  ({\mathbb Y}, {\cal B}({\mathbb Y}))$ we define    the conditional PM  of the  RV $Y$ conditioned on   $X$ by $
{\bf P}_{Y|X}(dy| X)\tri {\mathbb P}\big\{Y \in dy\big|X\big\}
  \equiv {\bf P}(dy|X)$, 
where $X$ is replaced by $x$ if     $X=x$. 
Conditional PMs are identified by stochastic kernels (mappings), $K:  {\cal B}({\mathbb Y }) \times {\mathbb X} \rightarrow [0,1]$ satisfying the following two properties \cite{dupuis-ellis1997}: \\
1) For every $x \in {\mathbb X}$, the set function $K(\cdot|x)$ is a probability measure
on $({\mathbb Y}, {\cal B}({\mathbb Y}))$, i.e.,  $K(\cdot|x)\in {\cal M}({\mathbb Y})$. \\
2) For every $A \in {\cal B}({\mathbb X})$, the function $K(A|\cdot)$ is
${\cal B}({\mathbb X})$-measurable. \\The set of such  maps $K(\cdot|\cdot)$ is  denoted by  ${\cal K}({\mathbb Y}|{\mathbb X})$.

  Given    three RVs,  $X: \Omega \rar {\mathbb X}$,  $Y: \Omega \rar {\mathbb Y}$, and $Z: \Omega \rar {\mathbb W}$, we 
   % defined on some probability space $(\Omega, {\cal F}, {\mathbb P})$. 
   say that RVs $(X, Y)$ are conditionally independent (CI) given RV $Z$ if  and only if  ${\bf P}_{X, Y|Z}={\bf P}_{X|Z}{\bf P}_{Y|Z}-\emph{a.s.}$,  equivalently,   ${\bf P}_{Y|Z,X}={\bf P}_{Y|Z}-\emph{a.s.}$; we denote the CI  by the Markov Chain (MC),  $X \leftrightarrow Z \leftrightarrow Y$.   %The specification $\emph{a.s.}$  (almost surely) is often omitted. 

The   mutual information between RVs $X$ and $Y$ conditioned on RV $Z$ is denoted by $I(X;  Y|Z)$ and  is defined by \cite{cover-thomas2006}
\begin{align}
 I(X;  Y|Z)  \tri 
   {\bf E} \Big\{  \log \Big(\frac{ d{\bf P}_{Y|X, Z}(\cdot|X,Z)}{d {\bf P}_{Y|Z}(\cdot|Z)}(Y)\Big)    \Big\} \label{cmi}
%% =
%   {\bf E} \Big\{  \log \Big(\frac{ d{\bf P}_{Y|X, Z}(\cdot|X,Z)}{d {\bf P}_{Y|Z}(\cdot|Z)} (Y)\Big)  \Big\}\in [0,\infty] \nonumber \\
%   &=  \int \log \Big(\frac{ d{\bf P}_{Y|X, Z}(\cdot|x,z)}{d {\bf P}_{Y|Z}(\cdot|z)}(y)\Big)  {\bf P}_{Y|X, Z}(dy|x,z) \otimes {\bf P}_{X, Z}(dx,dz). \label{cmi}
\end{align}
where the expectation is w.r.t. ${\bf P}_{X, Y, Z}(dx,dy,dz)$.
%= {\bf P}_{Y|X, Z}(dy|x,z)  {\bf P}_{X, Z}(dx,dz)$. 
By  \cite{pinsker1964,ihara1993},  $I(X;  Y|Z)\geq 0$ and $I(X;  Y|Z)= 0$ if and only if  ${\bf P}_{Y|X, Z}={\bf P}_{Y|Z}-$a.s.
The mutual  information between RVs $X$ and $Y$  is denoted by $I(X;  Y)$ and is a special case of  (\ref{cmi}), with $Z$ absent. The  entropy (differential) of  a RV $X$ with probability density function  $f_X(x)$ is  $H(X)\tri {\bf E}\big\{-\log \big(f_X(X)\big)\big\}=  -\int \log \big(f_X(x)\big) f_X(x) dx$.

\subsection{Signalling and Control in POMDPs}\label{sec:2.b}

Consider a partially observable Markov decision process (POMDP) defined over a finite horizon $t \in T_{+}^n\tri \{1,2,\ldots,n\}$:
\begin{align}
\{X_t, Y_t, A_t\}_{t=1}^n,\nonumber
\end{align}
where
\begin{itemize}
\item $X_t$ is the (unobservable) state variable,
\item $Y_t$ is the observation variable,
\item $A_t$ is the control (or action) variable.
\end{itemize}
The partially observable system is specified by stochastic kernels:
\begin{align}
\mathbb{P}(X_{t+1}|X_t,A_t) &= S_{t+1}(\cdot|x_t,a_t), \nonumber \\
\mathbb{P}(Y_t|X_t,A_t) &= Q_t(\cdot|x_t,a_t),\nonumber
\end{align}
and is penalized by a non-negative  measurable cost function $c_n: {\mathbb A}^n \times {\mathbb Y}^n \rar [0,\infty)$ which is additive
\begin{align}
c_n(a^n,x^n) = \sum_{t=1}^n \gamma_t(x_t,a_t).\nonumber
\end{align}

Next, we introduce the operational signalling and control protocol-a variation of  Shannon's operational channel capacity \cite{cover-thomas2006},  with the encoder replaced by a controller-encoder, as demonstrated in  Fig.~\ref{fig_2}.

\begin{definition}(Operational Signalling and Control of POMPDs)\label{def:code}  
Given the POMDP described in Section \ref{sec:2.b}, the signalling and control protocols shown in Fig.~\ref{fig_2} consist of:
\begin{itemize}
\item[(a)] A set of uniformly distributed messages $M^{(n)} : \Omega \rar  {\cal M}^{(n)}\tri \{1,2,\ldots, |{\cal M}^{(n)}|\} $,  where  $|{\cal M}^{(n)}|=card({\cal M}^{(n)})$;

\item[(b)] A set of  controller-encoder  strategies $g^n(\cdot)=\{g_1(\cdot), g_2(\cdot), \ldots, g_n(\cdot)\}$ defined by\footnote{The superscript in ${\bf E}^g$  indicates that the corresponding distribution ${\bf P}= {\bf P}^g$  depends the strategy $g$.} 
\begin{align}
&{\cal E}_{n}(\kappa) \triangleq  \Bigg\{g_t: {\cal M}^{(n)}  \times   {\mathbb A}^{t-1} \times {\mb Y}^{t-1}  \rar {\mathbb A}_t ~\text{s.t.}\nonumber\\
& A_1=g_1(M^{(n)}), A_2=g_2(M^{(n)},A_1,Y_1),\nonumber\\
&\ldots, A_n=g_n(M^{(n)}, A^{n-1}, Y^{n-1})\Bigg|\nonumber \\
& \mbox{(\ref{NDM-6_a})-(\ref{NDM-6-mc}) hold,} \hso   \frac{1}{n}   {\bf E}^g    \big\{c_n  (A^n, X^n) \big\} \leq \kappa  \Bigg\},  \label{NCM-6}
\end{align}
where (\ref{NDM-6_a})-(\ref{NDM-6-mc}) are the MCs
\begin{align}
& {\bf P}_{Y_t|Y^{t-1}, A^{t}, M^{(n)}}= {\bf P}_{Y_t|Y^{t-1}, A^{t-1}}, \label{NDM-6_a}\\  
 &{\bf P}_{A_t|A^{t-1}, Y^{t-1}, X^t}= {\bf P}_{A_t|A^{t-1}, Y^{t-1}}.  \label{NDM-6-mc} % \Longleftrightarrow \mbox{(MC),} \label{NDM-6-mc} 
% &    \mbox{(MC).}  \;X^t \leftrightarrow (A^{t-1},Y^{t-1}) \leftrightarrow A_t, \: \forall t \in {\mathbb Z}_+^n.
% \label{NDM-6-mc-mc}
\end{align}
\item[(c)] Decoder function,  $y^n \longmapsto d_{n}(y^n)\in  {\cal M}^{(n)}$,  with average  error  probability ${\bf P}_{error}^{(n)} = \frac{1}{|{\cal M}^{(n)}|} \sum_{m^{(n)}=1}^{|{\cal M}^{(n)}|} {\mathbb  P}\big\{d_n(Y^n) \neq   M^{(n)} \big|M^{(n)}=m^{(n)}\big\}= \epsilon_n$.
\end{itemize}
The \textit{CC  capacity} or \textit{signalling rate} of the POMDP system is $r_n\triangleq \frac{1}{n} \log |{\cal M}^{(n)}|$ bits/second. A CC capacity $R$ bits/second is called achievable if there exists a code sequence such that $\lim_{n\longrightarrow\infty} {\bf P}_{error}^{(n)}=0$ and $\liminf_{n \longrightarrow\infty}\frac{1}{n}\log{|{\cal M}^{(n)}|}\geq R$. {\it The operational definition of   CC capacity}  is the supremum of all achievable CC rates,  defined by  
\[
C^{FB}(\kappa) \triangleq \sup\{R : R \text{ is achievable}\}.
\]
\end{definition}
\par Each message in the set 
$  {\cal M}^{(n)}\tri \{1,2,\ldots, |{\cal M}^{(n)}|\} $ is  a  string of $\log |{\cal M}^{(n)}|$ bits. In digital transmission,   real-valued source  processes, $M_t: \Omega \rar {\mathbb R}^{n_m}, t=1,2, \ldots$ are first  quantized, and then represented by elements of  $  {\cal M}^{(n)}$.

\subsection{Information-Theoretic Characterization of $C^{FB}(\kappa)$}

In section, we demonstrate that the operational $C_{FB}(\kappa)$  (given in  Definition~\ref{def:code}), admits an equivalent information theoretic optimization problem, independently of $(g^n(\cdot)$, $d_n(\cdot), M^{(n)})$, and the constraint   $\lim_{n\longrightarrow\infty} {\bf P}_{error}^{(n)}=0$. 

{\it   Conditional Probability Measures (PMs) or Randomized Strategies.} Given the POMDP described in Section \ref{sec:2.b}, we introduce the larger set of conditional PMs or randomized strategies
\begin{align}
&{\cal P}_{n}(\kappa) \tri \Big\{ {\bf P}_{A_t|A^{t-1}=a^{t-1}, Y^{t-1}=y^{t-1}}={P}_t(da_t\big| a^{t-1},  y^{t-1}),  \nonumber \\
 &\hst \forall t \in \{1,2,\ldots,n\}\Big|\mbox{(\ref{NDM-6-mc}) }  \:\frac{1}{n} {\bf E}^{ P}\big\{ c_n(X^n,A^n)\big\} \leq \kappa   \Big\}.  \nonumber%\label{ci_ci}    
\end{align} 
\par The next theorem gives the information-theoretic characterization of the CC capacity  $C_{FB}(\kappa)$.

\begin{theorem} (Equivalent Information-Theoretic Characterization of $C_{FB}(\kappa)$) \\
\label{thm_capa}
Consider the POMDP described in Section \ref{sec:2.b}, and Definition~\ref{def:code}.
Moreover, given  $g \in {\cal E}_{n}(\kappa)$,  define the mutual information between $M^{(n)}$ and $Y^n$ by 
\begin{align}
I^g(M^{(n)} ; Y^n)  \tri  
  {\bf E}^g \Big\{  \log \Big(\frac{ {\bf P}_{Y^n| M^{(n)} }^g(dY^{n}\big| M^{(n)} )}{{\bf P}_{Y^n}^{g}(dY^{n})}\Big)    \Big\}.
\nonumber% \label{mi}
\end{align}
 For any   $P(\cdot|\cdot) \in {\cal P}_{n}(\kappa)$,  define the directed  information from  $A^n$ to $Y^n$   by 
\begin{align}
&I^P(A^n \rar  Y^n)  \tri   \sum_{t=1}^nI^P(A^t;Y_t|Y^{t-1}) \nonumber\\%\label{di_general_aa} \\
 &= \sum_{t=1}^n{\bf E}^P \Big\{  \log \Big(\frac{ {\bf P}_{Y_t|Y^{t-1}, A^t}^P(dY_t|Y^{t-1}, A^t)}{{\bf P}_{Y_t|Y^{t-1}}^{P}(dY_t|Y^{t-1})}\Big)    \Big\}
 \label{di_general}
\end{align}
where the conditional PMs,   ${\bf P}_{Y_t|Y^{t-1}, A^t}^P $ and ${\bf P}_{Y_t|Y^{t-1}}^P $, for any $t$ are obtained as follows  
\begin{align}
&{\bf P}_{Y_t|Y^{t-1}, A^t}^P = \int_{{\mathbb X}_t} Q_t(dy_t|x_t, A_t)    {\bf P}_t^P(dx_t|Y^{t-1}, A^{t-1}),   \label{cd_c1_c1}  \\
&{\bf P}_{Y_t|Y^{t-1}}^P = \int_{{\mathbb A}^t  \times {\mathbb X}_t}   { Q}_t(dy_t\big| x_t, a_t)    { P}_t(da_t|a^{t-1}, Y^{t-1}) \nonumber \\
& .{\bf P}_t^P(dx_t\big|a^{t-1}, Y^{t-1})
  {\bf P}_{t-1}^P(da^{t-1}|Y^{t-1}),  \forall t \in {\mathbb Z}_+^n. \label{cd_o1}
\end{align}
Define the optimization problems,  
% $ {C}_{FB}(\kappa)$
%where   
\begin{align}
&{C}_{FB}(\kappa) \tri \lim_{n \longrightarrow \infty} \frac{1}{n} {C}_{FB,n}(\kappa),  \label{CCIS_8_BC_D_n_a}\\
& {C}_{FB,n}(\kappa)\tri  \sup_{{\cal P}_{n}(\kappa)}  I^P(A^n \rar Y^n)  \label{CCIS_8_BC_D_n}
\end{align}
provided the limit and supremum exist.\\ 
Then the following coding theorems hold.\\
\textbf{(1) Converse theorem:} 
 If there exists a  code with ${\bf P}_{error}^{(n)}= \epsilon_n$  such that  $\lim_{n\longrightarrow\infty} {\epsilon}_n=0$,
then 
\begin{align} 
 &R \leq  \liminf_{n \longrightarrow\infty}\frac{1}{n}\log |{\cal M}^{(n)}| 
% \label{rate_A-LCM_n} \\
{\leq} \:  \liminf_{n \longrightarrow\infty}  \sup_{  {\cal E}_{n}(\kappa)}  \frac{1}{n}  I^{g}(M;  Y^n) 
%\label{CCIS_6a_BC_n}  \\
\nonumber \\
&{\leq} \:  \liminf_{n \longrightarrow \infty}   \sup_{{\cal P}_{n}(\kappa)} \frac{1}{n}I^P(A^n\rar Y^n)=  {C}_{FB}(\kappa). \label{CCIS_8_BC_nn}
\end{align}
\textbf{(2) Achievability theorem:} Any  CC rate     $R< C_{FB}(\kappa) $ is achievable, provided the directed information density $\frac{1}{n} \sum_{t=1}^n\log \big(\frac{ {\bf P}_{Y_t|Y^{t-1}, A^t}^P(dY_t|Y^{t-1}, A^t)}{{ P}_{Y_t|Y^{t-1}}^{{\bf P}}(dY_t|Y^{t-1})}\big)$ and the cost  $\frac{1}{n}c_{n}(A^n,X^n)$ are stable in the sense of \cite{ihara1993}, i.e.,  we can guarantee converge in probability to their mean values, as  $n \longrightarrow \infty$. 
%of the converse and direct CC Theorems in \cite{charalambous-kourtellaris-tziortzis:SICON-2024} hold, 
% Assumptions~\ref{ass:ati} (end of Section~\ref{sect:as-lim}) hold. 
% ~\ref{thm:converse-POSS}, \ref{thm:direct-POSS} hold. 
\end{theorem}
\begin{proof} \textbf{(1) Converse theorem:} Suppose $R$ is achievable, and hence  by Definition~\ref{def:code},   $\lim_{n\longrightarrow\infty} {\epsilon}_n=0$, 
$\liminf_{n \longrightarrow\infty}\frac{1}{n}\log|{\cal M}^{(n)}|\geq R$. To prove the statements we use  Fano's inequality   \cite{cover-thomas2006}. Recall that,  for each $n$,  $M^{(n)}=m^{(n)} \in {\cal M}^{(n)}$ is uniformly distributed with probability $1/|{\cal M}^{(n)}|$.   Now, we use  the identity  $I^g(M^{(n)};Y^n)=H(M^{(n)})-H^g(M^{(n)}|Y^n)$,  where $H^g(\cdot)$ denotes entropy and $H^g(\cdot|Y^n)$ entropy conditioned on $Y^n$,   and     Fano's inequality   \cite{cover-thomas2006},  $H^g(M^{(n)}|Y^n) \leq h(\varepsilon_n) + \varepsilon_n \log |{\cal M}^{(n)}|$, where   $h(z) \tri -z\log z -(1-z)\log (1-z), z \in [0,1]$,  to obtain   the  following  inequalities. 
\begin{align}
\frac{1}{n}\log & |{\cal M}^{(n)}| = \frac{1}{n}H(M^{(n)}) =\frac{1}{n} H^g(M^{(n)}|Y^n) \nonumber\\
& +\frac{1}{n} I^g(M^{(n)};Y^n),  \hst  \forall  g^n(\cdot) \in {\cal E}_{n}(\kappa) \nonumber \\
\leq & \frac{1}{n} h(\varepsilon_n) + \frac{1}{n}\varepsilon_n \log |{\cal M}^{(n)}|\nonumber\\
& + \frac{1}{n}I^g(M^{(n)}; Y^n), \hso  \forall  g^n(\cdot) \in {\cal E}_{n}(\kappa)   \nonumber\\%\label{fano_1}\\
\leq & \frac{1}{n}h(\varepsilon_n) + \frac{1}{n}\varepsilon_n \log |{\cal M}^{(n)}|  \label{fano_11} \\
& + \sup_{{\cal P}_{n}(\kappa)} \frac{1}{n}I^P(A^n\rar  Y^n), \; \forall g^n(\cdot)  \in {\cal E}_{n}(\kappa) \subseteq {\cal P}_{n}(\kappa)\nonumber
\end{align}
where to arrive in (\ref{fano_11})  we have used,   $I^g(M^{(n)}; Y^n)=\sum_{t=1}^n I^g(M^{(n)};Y_t| Y^{t-1})= \sum_{t=1}^n I^g(M^{(n)}, A^t;Y_t| Y^{t-1}), A_t=g_t(M^{(n)}, A^{t-1}, Y^{t-1})$, and the data processing inequality $I^g(M^{(n)}, A^t;Y_t| Y^{t-1})\leq I^g( A^t;Y_t| Y^{t-1}$ due to the MC, $M^{(n)} \leftrightarrow (A^{t},Y^{t-1}) \leftrightarrow Y_t, \forall t \in {\mathbb Z}_+^n$ (a consequence of code definition), and taking supremum over the larger set ${\cal P}_{n}(\kappa)$. From the above inequalities, 
 taking the supremum over ${\cal E}_{n}(\kappa)$, and   $\liminf_{n \longrightarrow \infty}$, using the fact that as, $n \longrightarrow \infty$, then  $\epsilon_n \longrightarrow 0$ and  $\frac{1}{n}h(\varepsilon_n)  \longrightarrow 0$, we deduce the stated inequalities. \textbf{(2) Achievability theorem:} The conditions are sufficient to apply the asymptotic equipartition property and random code generation, similar to  \cite[Theorem~4.5.1]{ihara1993}.
\end{proof}

%\begin{proof} 
%Assumptions~\ref{ass:ati} are sufficient  to  show (\ref{CCIS_8_BC_D_n_a}), via the   converse coding theorem  \cite{ihara1993,cover-pombra1989}, and the direct coding theorem using random coding arguments \cite{ihara1993,cover-pombra1989,kourtellaris-charalambousIT2015_Part_1}
%%
%%
%%prove the  channel coding theorems of noisy channels 
%%
%%
%% it follows that any achievable rate $R$ satisfies, $R \leq C_{FB}(\kappa)$ .  Second, by the direct coding theorem, using random coding arguments \cite{ihara1993,cover-pombra1989}, it follows that any rate $R < C_{FB}(\kappa)$ is achievable. 
%\end{proof}
%

\ \

 Sufficient conditions for existence of a maximum element in ${\cal P}_n(\kappa)$ of the   optimization problem 
 $C_{FB,n}(\kappa)$ are given in \cite{charalambous-stavrou2013aa} using weak convergence of probability measures for complete, separable metric spaces. 
 
% Further, \cite{charalambous-stavrou2013aa} also established that the payoff $I^P(A^n\rar Y^n)$ is concave over the space of distributions ${\cal P}_n(\kappa)$. This further implies that 
%  $C_{FB,n}(\kappa) $ is  concave, non-decreasing   in $\kappa \in (\kappa_{n,min}, \infty)\subseteq [0,\infty)$  \cite{charalambous-stavrou2013aa,charalambous-kourtellaris-tziortzis:SICON-2024}.

\ \

\begin{remark} (On Theorem~\ref{thm_capa}) In Theorem~\ref{thm_capa}, $I(A^n\rar Y^n)$  is  a nonlinear functional   of the conditional PMs    ${\bf P}_{Y_t|Y^{t-1}, A^t}^P$, ${\bf P}_{Y_t|Y^{t-1}}^{{ P}}$, $\forall t\in {\mathbb Z}_+^n$, that depend on the  randomized strategies  $P(\cdot|\cdot)\in {\cal P}_n(\kappa)$ and the a posteriori PMs  ${\bf P}_t^P(dx_t\big|a^{t-1}, y^{t-1}), {\bf P}_{t-1}^P(da^{t-1}|y^{t-1}),~  t \in {\mathbb Z}_+^n$. 
\end{remark}

 \section{Meta Dynamic Programming with Multiple Information States}
 \label{sect:meta-dp}
In this section, we develop DP recursions for the information-theoretic characterization of $C_{FB,n}(\kappa)$ in Theorem~\ref{thm_capa}.  We wish to identify sufficient statistics for the conditioning data $(A^{t-1}, Y^{t-1})$ appearing in the strategy $P_t(\cdot |A^{t-1}, Y^{t-1}), \forall t$. We proceed in several steps, by first 
introducing some assumptions.

\begin{assumptions}(Absolute Continuity)
\label{ass-1}
For abstract spaces, $Q_{t}( \cdot|x_t, a_t)\in {\cal M}({\mathbb Y}_{t})$  is  absolutely continuous w. r.t. the Lebesgue measure, namely,  there exists a  probability density function (PDF)  
 %  $ $s_{t+1}(x_{t+1}|x_t,  a_t),$  
   $q_{t}(y_t|x_t, a_t)$ such that 
\begin{align}
&Q_{t}(dy_t|x_t,  a_t)=q_{t}(y_t|x_t,a_t)dy_{t}, \; \forall t. \label{abs-con_1}
\end{align} 
For denumerable spaces (endowed with the discrete topology),   $S_{t+1}(\cdot |x_t, a_t)\in {\cal M}({\mathbb X}_{t+1})$, $Q_{t}( \cdot|x_t, a_t)\in {\cal M}({\mathbb Y}_{t})$  are replaced by transition probabilities  $s_{t+1}(x_{t+1}|x_t,  a_{t})\in {\cal M}({\mathbb X}_{t+1}),  q_{t}(y_t|x_t,  a_t)\in {\cal M}({\mathbb Y}_{t}), \forall t$.   
\end{assumptions}
 
 {\bf Step \#1:} In Theorem~\ref{theorem:aposteriori},  we show that 
 $\{{\bf P}^P_{t}(dx_{t}|A^{t-1}, Y^{t-1})\big| t\in  T_{+}^n\}$, which defines $ \{{\bf P}_{Y_t|Y^{t-1}, A^t}^P\big| t\in T_{+}^n\}$ in (\ref{cd_c1_c1})  
is an information state, it satisfies a recursion, and this recursion is Markov, which means it does not depend on $P(\cdot|\cdot) \in {\cal P}_n(\kappa)$.  
 
 \begin{theorem} (Recursion of  first a  posteriori PM) \label{theorem:aposteriori}
Consider the POMDP described in Section \ref{sec:2.b} with randomized strategies ${\cal P}_n(\kappa)$. Suppose  Assumptions~\ref{ass-1} hold\footnote{For denumerable spaces  integrals $\int_{{\mathbb X}_t}$  are replaced by sums $\sum_{x\in {\mathbb X}_t}$ and PMs by PMFs.}. Then the following statements hold.\\
{\bf (1)} The a posteriori PM  ${\bf P}^P_{t}(dx_{t}|A^{t-1}, Y^{t-1})$ is independent of  $P(\cdot|\cdot)$,  
\begin{align*}
 {\bf P}^P_{t}(dx_{t}|A^{t-1}, Y^{t-1})=\Xi_{t}[A^{t-1}, Y^{t-1}](dx_{t}),    \: \: \forall  P(\cdot|\cdot), \forall t
\end{align*} 
and  $\Xi_{t}[a^{t-1},y^{t-1}](dx_{t}) =\xi_{t}[a^{t-1}, y^{t-1}](dx_{t})$  satisfies the recursion
 \begin{align}
&\xi_{t+1}[a^{t}, y^{t}](dx_{t+1}) ={\bf T}_{t+1}\big(y_{i}, a_{t},\xi_{t}[a^{t-1}, y^{t-1}](\cdot)\big) (dx_{t+1})    \nonumber \\
&=\frac{{\bf T}_{t+1}^{un}\big(y_{i}, a_{t}, \xi_{t}[a^{t-1}, y^{t-1}](\cdot)\big)(dx_{t+1})    }{\int_{ {\mathbb X}_{t+1}}  {\bf T}_{t+1}^{un}\big(y_{t}, a_{t}, \xi_{t}[a^{t-1}, y^{t-1}](\cdot)\big)(dx_{t+1})}, \label{apost_1} \\
 & \xi_{1}(dx_{1})= {\bf P}(dx_1),  \label{apost_2} \\
&{\bf T}_{t+1}^{un}\big(y_{t}, a_{t},\xi_t[a^{t-1},y^{t-1}](\cdot)\big)(dx_{t+1}) \nonumber \\
&\tri\int_{{\mathbb X}_{t}} S_{t+1}(dx_{t+1}\big|y_{t}, x_{t},a_{t}) \nonumber \\
&\hst  . Q_{t}(dy_{t}\big|x_{t}, a_{t})    \xi_{t}[a^{t-1}, y^{t-1}](dx_t) .   
\end{align}
{\bf (2)} The  process $\{\Xi_{t}\big| t \in {\mathbb Z}_+^{n}\}$ that satisfies  the recursion  (\ref{apost_1})
 is an information state, i.e., a  Markov process.\\
{\bf (3)} The conditional PM ${\bf P}_{Y_t|Y^{t-1}, A^t}^P$ is independent of  $P(\cdot|\cdot)$,  
\begin{align}
&{\bf P}_{Y_t| Y^{t-1}=y^{t-1}, A^t=a^t}^P = {\bf P}_{t}(dy_{t}\big|a_t, \xi_t) \\
&= \int_{{\mathbb X}_t} Q_t(dy_t|x_t, a_t)    \; \xi_{t}[a^{t-1}, y^{t-1}](dx_t) , \; \forall P(\cdot|\cdot), \;  \forall t .
\label{cd_c1}  
\end{align}
\end{theorem}
\begin{proof} {\bf (1)}  
We prove the recursion by considering the Radon-Nikodym derivative (RND),
\begin{align}
&{\bf P}_{t+1}^P(dx_{t+1}\big| a^t, y^{t})\nonumber\\
&~~=\frac{{\bf P}_{t+1}^P(dx_{t+1},da_t, dy_t\big| a^{t-1}, y^{t-1})}{\int_{{\mathbb X}_{t+1}}{\bf P}_{t+1}^P(dx_{t+1},da_t, dy_t\big| a^{t-1}, y^{t-1})} \label{1st-is}
\end{align}
Because we have used the RND for abstract spaces, the only requirement for the derivation of the recursion is (\ref{abs-con_1}). Using the POMDP conditional PMs and the MC (\ref{NDM-6-mc}), we obtain the recursion,  and verify its independence on  $P(\cdot|\cdot)$ (i.e., the right-hand side of the recursion  (\ref{apost_1})  does not depend on strategies $P(\cdot|\cdot) \in {\cal P}_{n}(\kappa)$).
{\bf (2)}  We verify the   Markov property by considering the conditional PM of  $\Xi_{t+1}\tri{\bf P}_{t+1}(dx_{t+1}\big|A^{t}, Y^{t})$ conditioned on $\{\Xi^t, A^t\}$, and verifying it depends on $\{\Xi_{t}, A_t\}$, hence the Markov property holds.   
% Finally, since the right hand side of the recursion  (\ref{apost_1})  does not depend on    strategies $P(\cdot|\cdot) \in {\cal P}_{n}(\kappa)$, then    ${\bf P}^P_{t+1}(dx_{t+1}|a^{t}, y^{t})=  {\bf P}_{t+1}(dx_{t+1}|a^{t}, y^{t})$ does not depend on $P(\cdot|\cdot)$,   $ \forall t \in {\mathbb Z}_+^n$. 
{\bf (3)} Notice that $(a^t, y^{t-1})$ specifies $\xi_t[a^{t-1},y^{t-1}]$ and     by 
 reconditioning, 
\begin{align}
 &{\bf P}_t^P(dy_t\big|y^{t-1}, a^t)= {\bf P}_{t}^P(dy_{t}\big|a^t,  y^{t-1}, \xi_t)\nonumber\\
 &=\int_{{\mathbb X}_t} {\bf P}_t^P(dy_{t}\big|x_t, a^t, y^{t-1}, \xi_t)   {\bf P}_t^P(dx_t\big|   a^t, y^{t-1}, \xi_t)   \nonumber \\
&\sr{(a)}{=}\int_{{\mathbb X}_t} Q_t(dy_{t}\big|x_t, a_t)   \xi_{t}[a^{t-1}, y^{t-1}](dx_t) \nonumber% \label{if-1a}
  \end{align}
  where  $(a)$ follows from the definition of  the POMDP and the MC 
  ${\bf P}_t^P(dx_t\big|   a^t,  y^{t-1}, \xi_t)  ={\bf P}_t^P(dx_t\big| a^{t-1},  y^{t-1}, \xi_t) \equiv   \xi_{t}$. This implies ${\bf P}_t^P(dy_t\big|y^{t-1}, a^t)={\bf P}_{t}^P(dy_{t}\big|a^t,  y^{t-1}, \xi_t)={\bf P}_{t}(dy_{t}\big|a_t, \xi_t)$, establishing (\ref{cd_c1}).  
\end{proof}

{\bf Step \#2:} %Conditional PMs on the Space of  Conditional PMs (Information States).} 
We now focus on determining a tractable expression for $ \{{\bf P}_{Y_t|Y^{t-1}}^P\big| t\in  T_{+}^n\}$, which  appears in   $I^P(A^n \rar Y^n)$. 

We introduce the measurable space  $ ({\cal M}({\cal M}({\mathbb X}_t)),  {\cal B}({\cal M}({\mathbb X}_t)))$, where ${\cal M}({\cal M}({\mathbb X}_t))$ denotes the space of probability measures on the space of probability measures ${\cal M}({\mathbb X}_t)$, and 
${\cal B}({\cal M}({\mathbb X}_t))$ is the Borel $\sigma-$algebra generated by the weak topology on the space of probability measures  ${\cal M}({\mathbb X}_t)$, i.e., w.r.t.    the Prohorov metric. 

In Theorem~\ref{thm:sec-a-d}, we express $ {\bf P}_{Y_t|Y^{t-1}}^P$ w.r.t.  the a posteriori   PM
$\widehat{\Xi}[Y^{t-1}](d\xi_t)\tri {\bf P}_t^P(d \xi_t |Y^{t-1})\in   {\cal M}(   {\cal M}({\mathbb X}_t))$, i.e., the PM defined on  the space of PMs  $\Xi_t \equiv \Xi_t[A^{t-1}, Y^{t-1})(dx_t)\in  {\cal M}({\mathbb X}_t)$,   $\forall  t\in  T_{+}^n$. That is,  $\widehat{\xi}[y^{t-1}](d\xi_t)\tri {\bf P}_t^P(d\xi_t  |y^{t-1})$  and $\widehat{\xi}[\cdot](\cdot)\in {\cal K}({\cal M}(   {\cal M}({\mathbb X}_t))|{\mathbb Y}^{t-1})$ is a stochastic kernel.
% defined on the state space ${\mathbb S}_t \tri {\cal M}({\mathbb X}_t), \forall t$.  
  
We also   show that   $\widehat{\xi}[y^{t-1}](d\xi_t)\in   {\cal M}(   {\cal M}({\mathbb X}_t)),\forall  t\in T_{+}^n $  satisfies a recursion that depends on  strategies $\pi_t[\xi_t, y^{t-1}](da_t)\tri {\bf  P}_{t}(da_{t}\big|\xi_t, y^{t-1})$ instead of $P(da_t|a^{t-1}, y^{t-1}), \forall t$.   This means, the dependence of strategies $P(da_t|a^{t-1}, y^{t-1})$ on $a^{t-1}$ is via the information state $ \xi_t[a^{t-1},y^{t-1}](dx_t), \forall t$.

\begin{theorem} (A  posteriori PMs on the space of a  posteriori PMs) 
\label{thm:sec-a-d}
Consider the POMDP described in Section \ref{sec:2.b} with randomized strategies ${\cal P}_n(\kappa)$ and suppose  Assumptions~\ref{ass-1} hold. Moreover, consider the a posteriori PMs 
%$\xi_{t}[a^{t-1}, y^{t-1}](dx_{t})\tri {\bf P}_{t}(dx_{t}|a^{t-1}, y^{t-1})$, 
$\Xi_{t}[A^{t-1}, Y^{t-1}](dx_t) \in {\cal M}({\mathbb X}_t), \forall t$,  and define the a posteriori PMs
\begin{align*}
& \widehat{\Xi}_{t}^P[Y^{t-1}](d\xi_t)\tri  {\bf P}_{t}^P(d \xi_t| Y^{t-1})\in {\cal M}({\cal M}({\mathbb X}_t)), \; \forall t, \\
& \Pi_t[\Xi_t, Y^{t-1}](da_t)\tri {\bf  P}_{t}(da_{t}\big|\Xi_t, Y^{t-1})\in {\cal M}({\mathbb A}_t), \forall t
\end{align*}
Then, the following statements hold.\\
{\bf (1)} The  a posteriori PM $\big\{\widehat{\Xi}_{t}^P[Y^{t-1}](d{\xi}_{t})\big| t \in  T_{+}^n\big\}$  depends on $\big\{\Pi[\Xi_t, Y^{t-1}](da_t)\big| t \in  T_{+}^n\big\}$ (and not $\big\{P(da_t|A^{t-1}, Y^{t-1})\big| t \in  T_{+}^n\big\}$),
\begin{align}
\widehat{\Xi}_{t}^P[Y^{t-1}](d\xi_t)=\widehat{\Xi}_{t}^\Pi[Y^{t-1}](d{\xi}_{t}),~ \forall t\nonumber
\end{align}
and $\widehat{\Xi}_{t}^\Pi[Y^{t-1}](d{\xi_t})=\widehat{\xi}_{t}^\Pi[y^{t-1}](d{\xi_t})$  satisfies the recursion
 \begin{align}
&\widehat{\xi}_{t+1}^\pi[y^{t}](d{\xi}_{t+1}) \hst \: \forall t \in T_{+}^{n-1} \label{apost_1_pre}    \\
&=\widehat{\bf T}_{t+1}\big(y_t, \pi_{t}(\cdot\big|, \cdot, y^{t-1}) , \widehat{\xi}_{t}^\pi[y^{t-1}](\cdot) \big)(d\xi_{t+1})     \nonumber \\
&=\frac{\widehat{\bf T}_{t+1}^{un}\big(y_t, \pi_{t}(\cdot\big|, \cdot, y^{t-1})  ,  \widehat{\xi}_{t}^\pi[y^{t-1}](\cdot)\big)(d\xi_{t+1})    }{\int_{ {\cal M}({\mathbb X}_{t+1}) } \widehat{\bf T}_{t+1}^{un}\big(y_t, \pi_{t}(\cdot\big|, \cdot, y^{t-1})  ,  \widehat{\xi}_{t}^\pi[y^{t-1}](\cdot)\big)(d\xi_{t+1})     }, \nonumber \\
& {\bf P}_1^\pi(d\xi_1)=\delta_{{\bf P}_{X_1}}(d\xi_1),   
\nonumber
\end{align}
such that
\begin{align}
&\widehat{\bf T}_{t+1}^{un}\big(y_t, \pi_{t}(\cdot\big|\cdot, y^{t-1}) , \widehat{\xi}_{t}^\pi[y^{t-1}](\cdot)\big)(d{\xi}_{t+1})  \nonumber \\ 
 & \tri\int_{{\cal M}({\mathbb X}_t) \times {\mathbb A}_t} { \delta}_{ {\bf T}_{t+1}(y_{i}, a_{t}, \xi_t) }(d\xi_{t+1}){\bf P}_{t}(dy_{t}\big|a_t, \xi_t)\nonumber \\
& .  \pi_{t}(da_{t}\big|\xi_t, y^{t-1})  \widehat{\xi}_{t}^\pi[y^{t-1}](d\xi_t)  
\end{align}
where ${\bf P}_{t}(dy_{t}\big|a_t, \xi_t)$ is given by (\ref{cd_c1}), and where $\delta_{ {\bf T}_{t+1}(y_{t}, a_{t}, \xi_{t}) }(d\xi_{t+1})$  (resp.  $\delta_{{\bf P}_{X_1}}(d\xi_1)$) is the Dirac PM concentrated at ${\bf T}_{t+1}(y_{t}, a_{t}, \xi_{t}) $   (resp. ${\bf P}_{X_1}$).\\
%Moreover, $\widehat{\Xi}_{t+1}^P[y^{t}](\cdot)=\widehat{\Xi}_{t+1}^\Pi[y^{t}](d{\xi}_{t+1}), \forall t$, i.e., depend on $\Pi[\cdot](\cdot)$.\\
{\bf (2)}  The conditional PM given by  (\ref{cd_o1}) satisfies,  ${\bf P}_{Y_t|Y^{t-1}}^P={\bf P}_t^\Pi(dy_t|Y^{t-1})$,     and is expressed as 
\begin{align}
&{\bf P}_t^\pi(dy_t|y^{t-1}) =\int_{ {\cal M}({\mathbb X}_t)}  \Big(\int_{ {\mathbb A}_t  \times {\mathbb X}_t}   {\bf  P}_t(dy_t\big|  a_t, \xi_t)  \nonumber \\
& . \pi_t(da_t|\xi_t, y^{t-1}) 
\xi_t(dx_t)\Big)
  \widehat{\xi}_{t}^\pi[y^{t-1}](d\xi_t),  \forall t \in {\mathbb Z}_+^n. \label{cd_o1-n}
\end{align}
i.e.,  $\xi_t \in {\cal M}({\mathbb X}_t)$ is a realization averaged w.r.t. $\widehat{\xi}_{t}(d\xi_t)$.
\end{theorem}
\begin{proof} {\bf (1)} The conditional PM  of $\Xi_{t+1}$  given  $Y^t=y^t$ is determined by  using the RND,  
%\begin{align}
%&{\bf P}_{t+1}^P(d\xi_{t+1}\big| y^{t})=\frac{{\bf P}_{t+1}^P(d\xi_{t+1},dy_t\big| y^{t-1})}{\int_{{\cal M}({\mathbb X}_{t+1})}{\bf P}_{t+1}^P(d\xi_{t+1},dy_{t}\big|y^{t-1} )}, \label{der-rc_1_pr} \\
%& {\bf P}_{t+1}^P(d\xi_{t+1}, dy_t\big|dy^{t-1})\\
%&= \int_{{\cal M}({\mathbb X}_t)\times {\mathbb A}_t}{\bf P}_{t+1}^P(d\xi_{t+1},  da_{t},dy_t\big| dy^{t-1})\nonumber \\
%&=\int {\bf P}_{t+1}(d\xi_{t+1}\big|\xi_t, a_t, y^t) {\bf P}_{t}^P(dy_{t}\big|a_t, \xi_t, y^{t-1})\nonumber \\
%& .  {\bf  P}_{t}^P(da_{t}\big|\xi_t, y^{t-1})  {\bf P}_{t}^P(d\xi_t|y^{t-1}) 
%\label{der-rc_2_pr-new1}\\
%&=\int \delta_{ {\bf T}_{t+1}^{nor}(y_{i}, a_{t}, \xi_t) }(d\xi_{t+1}){\bf P}_{t}(dy_{t}\big|a_t, \xi_t)\nonumber \\
%& .  \Pi_{t}(da_{t}\big|\xi_t, y^{t-1})  \widehat{\Xi}_{t}^P[y^{t-1}](d\xi_t) 
%\label{der-rc_2_pr-new2}
%\end{align}
\begin{align}
&{\bf P}_{t+1}^P(d\xi_{t+1}\big| y^{t})=\frac{{\bf P}_{t+1}^P(d\xi_{t+1},dy_t\big| y^{t-1})}{\int_{{\cal M}({\mathbb X}_{t+1})}{\bf P}_{t+1}^P(d\xi_{t+1},dy_{t}\big|y^{t-1} )}, \label{der-rc_1_pr} \\
& {\bf P}_{t+1}^P(d\xi_{t+1}, dy_t\big|dy^{t-1})\\
&= \int_{{\cal M}({\mathbb X}_t)\times {\mathbb A}_t}{\bf P}_{t+1}^P(d\xi_{t+1}, d\xi_t, da_{t},dy_t\big| dy^{t-1})\nonumber \\
&=\int_{{\cal M}({\mathbb X}_t)\times {\mathbb A}_t} {\bf P}_{t+1}(d\xi_{t+1}\big|\xi_t, a_t, y^t) {\bf P}_{t}^P(dy_{t}\big|a_t, \xi_t, y^{t-1})\nonumber \\
& .  {\bf  P}_{t}^P(da_{t}\big|\xi_t, y^{t-1})  {\bf P}_{t}^P(d\xi_t|y^{t-1}) 
\label{der-rc_2_pr-new1}
\end{align}
\begin{align}
&=\int_{{\cal M}({\mathbb X}_t)\times {\mathbb A}_t} \delta_{ {\bf T}_{t+1}(y_{i}, a_{t}, \xi_t) }(d\xi_{t+1}){\bf P}_{t}(dy_{t}\big|a_t, \xi_t)\nonumber \\
& .  \pi_{t}(da_{t}\big|\xi_t, y^{t-1})  \widehat{\xi}_{t}^P[y^{t-1}](d\xi_t) 
\label{der-rc_2_pr-new2}
\end{align}
where (\ref{der-rc_2_pr-new2})  is  due  to the  recursion of Theorem~\ref{theorem:aposteriori},  which implies ${\bf P}_{t+1}(d\xi_{t+1}\big|\xi_t, a_t, y^t)$ is the Dirac PM,  the  re-conditioning property, which implies 
\begin{align}
 & {\bf P}_{t}^P(dy_{t}\big|a_t, \xi_t, y^{t-1})=\int_{{\mathbb X}_t} {\bf P}_t^P(dy_{t}\big|x_t, a_t, \xi_t, y^{t-1})   \nonumber \\
&. {\bf P}_t^P(dx_t\big|   a_t, \xi_t, y^{t-1})   \nonumber \\
&\sr{(a)}{=}\int_{{\mathbb X}_t} Q_t(dy_{t}\big|x_t, a_t)   \xi_{t}[a^{t-1}, y^{t-1}](dx_t)  \label{if-1}
  \end{align}
  where  $(a)$ follows from  the definition of  the POMDP and  
  ${\bf P}_t^P(dx_t\big|   a_t, \xi_t, y^{t-1})  ={\bf P}_t^P(dx_t\big|  \xi_t, y^{t-1}) =  \xi_{t}$. 
From the above we obtain the recursion.  Using induction, we can show 
 the recursion satisfies   $\widehat{\xi}_{t+1}^P[y^{t}](\cdot)=\widehat{\xi}_{t+1}^\pi[y^{t}](d{\xi}_{t+1}), \forall t$, i.e., depend on $\pi[\cdot](\cdot)$.
{\bf (2)}  Consider
\begin{align}
&{\bf P}_t^P(dy_t\big|y^{t-1}) =  \int_{ {\cal M}({\mathbb X}_t)\times  {\mathbb A}_t  \times {\mathbb X}_t}   {\bf  P}_t^P(dy_t\big| x_t, a_t, \xi_t, y^{t-1})  \nonumber \\
& . {\bf  P}_t^P(da_t|x_t, \xi_t, y^{t-1}) 
{\bf P}_t^P(dx_t\big|\xi_t, y^{t-1})
  {\bf P}_{t}^P(d\xi_t\big|y^{t-1}).  
 \label{cd_o1-nn}
\end{align}
By substituting into (\ref{cd_o1-nn}), ${\bf  P}_t^P(dy_t\big| x_t, a_t, \xi_t, y^{t-1})=Q_t(dy_t\big| x_t, a_t)$, ${\bf  P}_t^P(da_t|x_t, \xi_t, y^{t-1}) ={\bf P}_t(da_t|\xi_t, y^{t-1}) \tri \pi_t[\xi_t, y^{t-1}](da_t)$, ${\bf P}_t^P(dx_t\big|\xi_t, y^{t-1})=\xi_t,$  ${\bf P}_{t}^P(d\xi_t\big|y^{t-1})=\widehat{\xi}_t^\pi[y^{t-1}](d\xi_t)$,  we obtain (\ref{cd_o1-n}).
\end{proof}

{\bf Step \#3:} % Equivalent Expressions of Payoff   and Average Cost.}  
In  the next theorem,  we  invoke Theorem~\ref{theorem:aposteriori} and Theorem~\ref{thm:sec-a-d} to express the payoff $I^P(A^n \rar Y^n)$ and average cost in terms of $\widehat{\Xi}_{t}^\Pi[Y^{t-1}](\cdot) \in {\cal M}({\cal M}({\mathbb X}_t))$ and the strategies $\Pi_t[\xi_t, y^{t-1}](da_t)\tri {\bf  P}_{t}(\cdot\big|\xi_t, y^{t-1}) \in {\cal M}({\mathbb A}_t), \forall t \in T_{+}^n$. Then we derive an alternative equivalent information-theoretic characterization of $C_{FB,n}(\kappa)$. 

\begin{theorem} (Equivalent payoff and average cost) 
\label{thm:equiv}
Consider the POMDP described in Section \ref{sec:2.b} with randomized strategies ${\cal P}_n(\kappa)$ and suppose  Assumptions~\ref{ass-1} hold. \\
Then the payoff $I^P(A^n \rar Y^n)$  and average cost  $\frac{1}{n} {\bf E}^{ P}\big\{ c_n(X^n,A^n)\big\}$  in 
Theorem~\ref{thm_capa} is expressed as follows. \\
{\bf Payoff:}
\begin{align}
&I^P(A^n \rar  Y^n)  =   \sum_{t=1}^nI^{\Pi}(A_t, \Xi_t;Y_t|Y^{t-1}) \label{di_general_aa} \\
 &= \sum_{t=1}^n{\bf E}^\Pi \Big\{  \log \Big(\frac{ {\bf P}_t(dY_t| A_t,\Xi_t)}{{\bf P}_t^{^\Pi}(dY_t|Y^{t-1})}\Big)    \Big\}
 \label{di_general-eqv}
\end{align}
where ${\bf P}_t(dy_t\big|A_t, \Xi_t)={\bf P}_{Y_t|A_t, \Xi_t}$ given by (\ref{cd_c1})  and   ${\bf P}_t^\Pi(dy_t\big|Y^{t-1})={\bf P}_{Y_t|Y^{t-1}}^\Pi$ given by (\ref{cd_o1-n}).\\
{\bf Average Cost Constraint:}
\begin{align}
&{\bf E}^{ P}\Big\{ c_n(X^n,A^n)\Big\}={\bf E}^{ P}\Big\{\sum_{t=1}^n  \gamma_t(X_t,A_t)\Big\}  \label{av-cost}  \\
&={\bf E}^{\Pi}\Big\{\sum_{t=1}^n \int_{{\mathbb X}_t \times {\mathbb A}_t} \gamma_t(x_t,a_t)  \Pi_t(da_t\big|\Xi_t, Y^{t-1}) \nonumber \\
&\hst .\Xi_t[A^{t-1}, Y^{t-1}](dx_t) \Big\} \leq \kappa . \label{av-cost-eqv}
\end{align}
{\bf Equivalent Optimization Problem $C_{FB,n}(\kappa)$:} The optimization problem $C_{FB,n}(\kappa)$ is equivalent to 
\begin{align}
&C_{FB,n}^\Pi(\kappa)=\sup_{{\cal P}_n^{\Pi}(\kappa)} \sum_{t=1}^n{\bf E}^P \left\{  \log \Big(\frac{ {\bf P}_t(dY_t| A_t,\Xi_t)}{{\bf P}_t^{^\Pi}(dY_t|Y^{t-1})}\Big)    \right\},  \label{cap_general-eqv}
\end{align}
where
\begin{align}
{\cal P}_n^\Pi(\kappa)=\Big\{\Pi_t(da_t\big|\Xi_t, Y^{t-1}), \forall t\in  T_{+}^n \big| \: \mbox{(\ref{av-cost-eqv}) holds}\Big\}. \nonumber 
\end{align}
\end{theorem}
\begin{proof} Substituting into the directed information  (\ref{di_general}), ${\bf P}_{Y_t|Y^{t-1}, A^t}^P={\bf P}_{Y_t|A_t, \Xi_t}={\bf P}_t(dy_t\big|A_t, \Xi_t)$ given by (\ref{cd_c1})  and   ${\bf P}_{Y_t|Y^{t-1}}^P={\bf P}_t^\Pi(dy_t\big|Y^{t-1})$ given by (\ref{cd_o1-n}) we obtain  (\ref{di_general-eqv}). The equivalent average cost (\ref{av-cost-eqv}) follows by re-conditioning and using  MC  (\ref{NDM-6-mc}). Putting these together, we 
obtain (\ref{cap_general-eqv}). 
% ${\bf E}^{ P}\big\{ \gamma_t(X_t,A_t)\big\} ={\bf E}^{ P}\big\{{\bf E}^{ P}\big\{ \gamma_t(X_t,A_t) \big|A^{t-1}, Y^{t-1}\big\}  \Big\}$,  where 
%\begin{align}
%& {\bf E}^{ P}\big\{ \gamma_t(X_t,A_t) \big|A^{t-1}, Y^{t-1}\big\} \\
%&=\int  \gamma_t(x_t,a_t) {\bf P}_t^P(dx_t, da_t\big| A^{t-1}, Y^{t-1})\\
%&=\int  \gamma_t(x_t,a_t) {\bf P}_t^P(dx_t, da_t\big| A^{t-1}, Y^{t-1}), \\
%&{\bf P}_t^P(da_t\big| X_t, A^{t-1}, Y^{t-1}){\bf P}_t^P(dx_t\big| A^{t-1}, Y^{t-1})\\
%&{\bf P}_t^P(da_t\big| X_t, A^{t-1}, Y^{t-1})\sr{(a)}{=}
%{P}_t(da_t\big| A^{t-1}, Y^{t-1})
%\end{align}
%and where $(a)$ follows by MC  (\ref{NDM-6-mc-mc}). 
\end{proof}

\ \

{\bf Step \#4:}  In  Theorem~\ref{thm:eqv-dp1},  we present a more advanced version of DP referred to as {\it meta  DP} that characterizes  $C_{FB,n}^\Pi(\kappa)$.

First, we recall that  $C_{FB,n}^\Pi(\kappa)$ is a  convex optimization problem \cite{charalambous-stavrou2013aa}. Hence, to develop  the meta DP approach we  introduce the  associated  value process
using 
 the Lagrangian of  the constraint ${\cal P}_n^\Pi(\kappa)$.

\begin{definition}(Meta  Value Process)
\label{def:hlvp}
Consider  \eqref{cap_general-eqv} in Theorem~\ref{thm:equiv}.  Then, the value process or  optimal cost-to-go ${\cal V}_{t}(\cdot): {\mathbb Y}^{t-1} \rar [0,\infty)$, over  $\{t, \ldots, n\}$  when the optimal strategy $\Pi(\cdot\big|\cdot) =\Pi^o(\cdot\big|)   \in {\cal P}_{t-1}^\Pi(\kappa)$ is  used over $\{1,2, \ldots, t-1\}$, conditioned  on any  $Y^{t-1}=y^{t-1}$   is   
\begin{align}
 & {\cal V}_t^\lambda(y^{t-1}) \tri   \inf_{\{(\Xi_j, y^{j-1}) \mapsto \Pi_j(\cdot|\Xi_j, y^{j-1})  \in {\cal M}_{j}({\mathbb A}_j)\}_{j=t}^n }  {J}_{t,n}^{\Pi,\lambda}(y^{t-1}),    
  \label{opt-ctg-g}\\
&{J}_{t,n}^{\Pi,\lambda}(y^{t-1}) 
 \tri   {\bf  E}^{\Pi} \Big\{  \sum_{j=t}^n I^{\Pi}(A_j, \Xi_j;Y_j|Y^{j-1})\nonumber \\
 &-\lambda \sum_{j=t}^n  \gamma_j(X_j,A_j) \Big|Y^{t-1}=y^{t-1} \Big\},  \hst \forall t \in  T_{+}^n \label{opt-ctg-pbp-tc} 
% \\
% &=\inf_{  \gamma^{k}\in {\cal P}_{t,n}^{-k}}  {\mathbb E}^{^{\gamma^k, \gamma^{-k,o}}}    \Big\{ \nonumber \\
%&\hst
%\sum_{j=t}^{n} \int_{{\mathbb X}_j\times {\mathbb L}_j^{-k}}   \ell(j, x_j,\gamma_j^k(I_j^k),  \gamma_j^{-k,o}(\Delta_j^{(K)},\lambda
%_j^{-k}))\nonumber \\
%&. {\bf P}_j^{{\gamma^k, \gamma^{-k,o}}} (dx_j, d\lambda_j^{-k}\; \big|I_j^k) \Big| \; I_t^k= i_t^k  \Big\}. \nonumber 
\end{align}
where $\lambda \in [0,\infty)$ is the Lagrangian associated with the average cost constraint (if there is no cost then $\lambda=0$). 
\end{definition}

Next, we state the main theorem that gives the meta DP equations and their consequences.

\begin{theorem}(Meta DP-necessary and sufficient conditions)
\label{thm:eqv-dp1} Consider the Definition \ref{def:hlvp}. Then the following statements hold.\\
{\bf (1)} The value process of Definition~\ref{def:hlvp} is expressed  w.r.t.  to the high level state $\widehat{\xi}_{t}^\pi[y^{t-1}](\cdot) \in {\cal M}({\cal M}({\mathbb X}_t))$,   as follows 
\begin{align}
& {\cal V}_{n}^\lambda(y^{n-1}) = \sup_{(\xi_n, y^{n-1}) \mapsto   \pi_n(\cdot|\xi_n, y^{n-1})\in {\cal M}({\mathbb A}_n)}    \biggl\{    \nonumber\\  
& \int_{{\cal M}({\mathbb X}_n)} \bigg(\int_{    {\mathbb Y}_n\times {\mathbb A}_n}  \log \left(\frac{ {\bf P}_n(dy_n| a_n,\xi_n)}{{\bf P}_n^{\pi}(dy_n|y^{n-1})}\right)   \nonumber \\
& .{\bf P}_n(dy_n\big| a_n, \xi_n) \pi_n(da_n\big|\xi_n, y^{n-1}) \bigg)-\lambda  \int_{{\mathbb X}_n \times {\mathbb A}_n} \gamma_n(x_n,a_n)  \nonumber \\
  &. \pi_n(da_n\big|\xi_n, y^{n-1}) \xi_n(dx_n)  \bigg)
  \widehat{\xi}_{n}^{\pi}[y^{n-1}](d\xi_n) \Bigg\},   \label{v-pbp-a} \\
& {\cal V}_{t}^\lambda(y^{t-1})= \sup_{(\xi_j, y^{j-1}) \mapsto     \pi_j(\cdot|\xi_j,y^{j-1})\in {\cal M}({\mathbb A}_j), j=t, \ldots, n}  \bigg\{ \nonumber\\
&\hso \sum_{j=t}^n   \int_{{\cal M}({\mathbb X}_j)} \bigg(  \int_{    {\mathbb Y}_j\times {\mathbb A}_j}  \log \left(\frac{ {\bf P}_j(dy_j| a_j,\xi_j)}{{\bf P}_j^{\pi}(dy_j|y^{j-1})}\right)   \nonumber \\
& .{\bf P}_j(dy_j\big| a_j, \xi_j) \pi_j(da_j\big|\xi_j, y^{j-1}) )-\lambda  \int_{{\mathbb X}_j\times {\mathbb A}_j} \gamma_j(x_j,a_j)  \nonumber \\
  &. \pi_j(da_j\big|\xi_j, y^{j-1}) \xi_j(dx_j)  \bigg)
  \widehat{\xi}_{j}^{\pi}[y^{j-1}](d\xi_j \bigg\}, \hso \forall t \in  T_{+}^{n-1}.   \label{v-pbp}  
\end{align}
{\bf (2) Necessary Conditions of Optimality:}  Suppose an  optimal  $P(\cdot|\cdot )\in {\cal P}_n(\kappa)$  exists for  $C_{FB,n}(\kappa)$. \\
(2.1)  The value process 
${ V}_t^\lambda(Y^{t-1}),  \forall  t\in T_{+}^n $ 
   satisfies   the   DP   equations,  (\ref{dp_1_nnn_1}),   (\ref{dp_1_na}):  $\forall \; Y^{t-1}=y^{t-1}$, $\forall \;  \widehat{\Xi}_{t}^\Pi[Y^{t-1}](\cdot) =\widehat{\xi}_{t}[y^{t-1}](\cdot) \in {\cal M}({\cal M}({\mathbb X}_t))$,
\begin{align}
&{\cal V}_n^\lambda(y^{t-1})=\mbox{(\ref{v-pbp-a})}  \label{dp_1_nnn_1}\\
&{\cal V}_t^\lambda(y^{t-1}) = \sup_{(\xi_t, y^{t-1}) \mapsto     \pi_t(\cdot|\xi_t,y^{t-1})\in {\cal M}({\mathbb A}_t)}    \bigg\{ \nonumber \\
&  \int_{{\cal M}({\mathbb X}_t)} \bigg(  \int_{    {\mathbb Y}_t\times {\mathbb A}_t}  \log \left(\frac{ {\bf P}_t(dy_t| a_t,\xi_t)}{{\bf P}_t^{\pi}(dy_t|y^{t-1})}\right)   \nonumber \\
& .{\bf P}_t(dy_t\big| a_t, \xi_t) \pi_t(da_t\big|\xi_t, y^{t-1}) )-\lambda  \int_{{\mathbb X}_t\times {\mathbb A}_t} \gamma_t(x_t,a_t)  \nonumber \\
  &. \pi_t(da_t\big|\xi_t, y^{t-1}) \xi_t(dx_t)  \bigg)
  \widehat{\xi}_{t}[y^{t-1}](d\xi_t)\nonumber \\
  & +  \int_{{\mathbb Y}_{t}}  {\cal  V}_{t+1}^\lambda(y^t)  {\bf P}_t^{\pi}(dy_t\big|y^{t-1})    \bigg\}, \hso \forall t \in T_{+}^{n-1}
 \label{dp_1_na-1}
\end{align}
where in  (\ref{dp_1_na-1}) 
 the   conditional   the conditional  PM ${\bf P}_t^\pi(dy_t\big|y^{t-1})$ is given by  (\ref{cd_o1-n}). \\
(2.2) The supremum in the DP recursions, occurs at $\pi_t^{o}[\xi_t, y^{t-1}](da_t)=\pi_t^{o,sep}[\xi_t, \widehat{\xi}_t](da_t), \forall t$, that is, strategies are separated (they depend on $y^{t-1}$ through $\widehat{\xi}_t$). \\
{\bf (3) Verification-Sufficient  Conditions of  Optimality:} \\
(3.1)  If the value process  ${\cal  V}_t^\lambda(\cdot),  \forall t \in T_{+}^n$  satisfies the DP equations  \eqref{v-pbp-a}, \eqref{v-pbp},  then  (almost surely), 
%  $\forall k \in {\mathbb Z}_+^K$,
\begin{align}
 {\cal V}_t^\lambda(Y^{t-1})    \leq J_{t,n}^{\Pi,\lambda}(Y^{t-1}),  \hso   \forall \Pi, \:  \forall t  \label{in_1_g} 
\end{align}
and the resulting supremum of the DP equations, $\pi_t^o[\xi_t, y^{t-1}](\cdot)\in {\cal M}({\mathbb A}_t), \forall t$, is  optimal.\\
(3.2) Suppose $\pi_t^{o}[\xi_t, y^{t-1}](da_t)=\pi_t^{,sep,o}[\xi_t, \widehat{\xi}_t](da_t)$
 is a strategy 
such that,  for all $\{y^{t-1}, \widehat{\xi}_t\}$ achieves the supremum in the DP equations \eqref{v-pbp-a}, \eqref{v-pbp} for $t=1, \ldots, n$. Then $\pi_t^{sep,o}[\xi_t, \widehat{\xi}_t](da_t)$ is  optimal and ${\cal  V}_t^\lambda(Y^{t-1}) = J_{t,n}^{\Pi^o,\lambda}(Y^{t-1}), \forall t$  (almost surely).
\end{theorem}
\begin{proof} {\bf (1)} This is immediate from Definition \ref{def:hlvp}. {\bf (2.1)} This follows from Theorem~\ref{thm:equiv} and Definition~\ref{def:hlvp} by performing the conditional expectation, and breaking $\sum_{j=t}^n$ into the current and future cost. {\bf (2.2)} At time $t=n$, by (\ref{dp_1_nnn_1}), the supremum occurs at  $\pi_n^{o}[\xi_n, y^{n-1}](da_n)=\pi_n^{o,sep}[\xi_n, \widehat{\xi}_n](da_n)$. This implies (\ref{cd_o1-n}), ${\bf P}_n^{pi^o}(dy_n|y^{n-1})=
 {\bf P}_n^{\pi^{o,sep}}(da_n|\xi_n, \widehat{\xi}_n)$, and consequently   ${\cal V}_n^\lambda(y^{t-1})={V}_t^\lambda(\widehat{\xi}_n)$. Using induction we can also show $\pi_t^{o}[\xi_t, y^{t-1}](da_t)=\pi_t^{o,sep}[\xi_n, \widehat{\xi}_t](da_t), \forall t \in T_{+}^{n-1}$.
 {\bf (3)} This can be shown following variations of standard DP derivations \cite{kumar-varayia1986}. 
\end{proof}

By Theorem~\ref{thm:eqv-dp1}, {\bf (2.2)},  we can  obtain a significant simplification of the  meta  DP equations by considering the recursion of Theorem~\ref{thm:sec-a-d} corresponding to   fully  separated strategies, i.e., 
$\Pi_t^{sep}[\Xi_t, \widehat{\Xi}_t^{\Pi^{sep}}](da_t)\tri {\bf  P}_{t}(da_t\big|\Xi_t, \widehat{\Xi}_t^{\Pi^{sep}}) \in {\cal M}({\mathbb A}_t), \forall t \in T_{+}^n$. A sufficient (but not necessary) condition for such strategies is that the $\sigma-$algebra generated by $\{Y_t\big|t\in T_{+}^n\}$ is the same as that of $\{\widehat{\Xi}_t^\Pi\big|t\in T_{+}^n\}$.

\begin{theorem}(Meta DP with fully separated strategies-necessary and sufficient conditions)
\label{thm:eqv-dp2}
Consider Theorem~\ref{thm:eqv-dp1} and restrict the set of strategies $\Pi\in {\cal P}_n^\Pi(\kappa)$ to fully separated strategies $\Pi_t^{sep}[\Xi_t, \widehat{\Xi}_t^{\Pi^{sep}}](da_t) \in {\cal M}({\mathbb A}_t), \forall t \in T_{+}^n$.\\
{\bf (1)} The value process  ${\cal V}_t^\lambda(\cdot)$   of Theorem~\ref{thm:eqv-dp1}, given by  (\ref{v-pbp-a}),  (\ref{v-pbp}) satisfies, 
\begin{align}
{\cal V}_t^\lambda(Y^{t-1})={V}_t^\lambda(\widehat{\Xi}_t^{\Pi^{sep}})-a.s.,  \hso \forall t
%\forall y^{t-1}. 
\end{align}
with ${\bf P}_t^{\Pi}(dy_t|Y^{t-1})={\bf P}_t^{\Pi^{sep}}(dy_t|\widehat{\Xi}_t^{\Pi^{sep}}), \forall t$, and  
 the value function ${V}_t^\lambda(\cdot)$ satisfies  the DP equations $\forall \; \widehat{\Xi}_{t}^{\Pi^{sep}}[Y^{t-1}](\cdot)=\widehat{\xi}_{t}[y^{t-1}](\cdot) \in {\cal M}({\cal M}({\mathbb X}_t))$:
\begin{align}
& {V}_{n}^\lambda(\widehat{\xi}_n) =\mbox{(\ref{v-pbp-a})}  \label{dp_1_nnn_1-n}  \\
&{ V}_t^\lambda{\xi}_t) = \sup_{(\xi_t,y^{t-1}) \mapsto     \pi_t(\cdot|\xi_t,y^{t-1})\in {\cal M}({\mathbb A}_t)}    \bigg\{ \nonumber \\
&  \int_{{\cal M}({\mathbb X}_t)} \bigg(  \int_{    {\mathbb Y}_t\times {\mathbb A}_t}  \log \left(\frac{ {\bf P}_t(dy_t| a_t,\xi_t)}{{\bf P}_t^{\pi}(dy_t|\widehat{\xi}_t)}\right)   \nonumber \\
& .{\bf P}_t(dy_t\big| a_t, \xi_t) \pi_t(da_t\big|\xi_t,y^{t-1}) )-\lambda  \int_{{\mathbb X}_t\times {\mathbb A}_t} \gamma_t(x_t,a_t)  \nonumber \\
  &. \pi_t^{sep}(da_t\big|\xi_t,y^{t-1}) \xi_t(dx_t)  \bigg)
  \widehat{\xi}_{t}[y^{t-1}](d\xi_t)\nonumber \\
  & +  \int_{{\mathbb Y}_{t}}  {  V}_{t+1}^\lambda\big(\widehat{\bf T}_{t+1}(y_t, \pi_{t}(\cdot\big|\cdot, \widehat{\xi}_t) , \widehat{\xi}_{t}(\cdot)  \bigg) \nonumber \\
  &. {\bf P}_t^{\pi}(dy_t\big| \widehat{\xi}_t  )    \bigg\}, \hso \forall t \in T_{+}^{n-1}
 \label{dp_1_na}
\end{align}
where  the conditional  PM ${\bf P}_t^{\Pi^{sep}}(dy_t\big|\widehat{\xi}_t)$ is given by  (\ref{cd_o1-n}) with $\Pi_{t}(\cdot\big|\xi_t, y^{t-1})$ replaced by $  \Pi_{t}^{sep}(\cdot\big|\xi_t,  \widehat{\xi}_t), \forall t$. \\
{\bf (2) Verification-Sufficient  Conditions of  Optimality:} \\
(2.1)  If the value process  ${  V}_t^\lambda(\cdot),  \forall t \in {1,2,\ldots,n}$  satisfies the DP equation \eqref{dp_1_nnn_1-n}, \eqref{dp_1_na},  then  (almost surely), 
\begin{align}
 {V}_t^\lambda(\widehat{\Xi}_t^{\Pi^{sep}}))    \leq J_{t,n}^{\Pi,\lambda}(Y^{t-1}) ,  \;\forall \Pi(\cdot|\cdot), \:  \forall t  \label{in_1_g-a} 
\end{align}
and the resulting $\Pi^{sep}(\cdot|\cdot)$ is  optimal.\\
(2.2) Suppose $\Pi_t^{o,sep}[\xi_t, \widehat{\xi}_t](da_t)$
 is a strategy 
such that,  for all $\widehat{\xi}_t$ achieves   the supremum  in the DP eqns of part (1)  for $t=1, \ldots, n$. \\
Then $\Pi_t^{o,sep}[\xi_t, \widehat{\xi}_t](da_t)$ is optimal and ${ V}_t^\lambda(\widehat{\Xi}_t^{\Pi^{o, sep}}) = J_{t,n}^{\Pi^{o,sep},\lambda}(Y^{t-1}), \forall t$  (almost surely).
\end{theorem}
\begin{proof} The proof is similar  to  
Theorem~\ref{thm:eqv-dp1}.
\end{proof}

We conclude this section with the following remark.
\begin{remark} We relate $C_{FB,n}(\kappa)$ to the optimal cost of POMDPs subject to a rate constraint. 

{\bf Cost-Rate of POMDPs:}  Consider the cost-rate   problem,    % with the constraint and payoff interchanged, 
%defined by 
 \begin{align*}
\kappa_{n}(C) \triangleq   \inf_{{\cal P}_n, \: \mbox{s.t.} \:  \frac{1}{n} I^P(A^n \rar Y^n) \: \geq \: C} 
 {\bf  E}^P\big\{c_{n}(A^n, X^{n})\big\} \label{cap_fb_1_TC_n_n}
 \end{align*}
 where  $C \in [0, \infty]$. 
 Clearly, we can obtain analogous meta DP equations for the cost rate.   \\
 {\bf Relation to Optimal Cost of  POMDPs:}
 Let ${\cal E}_{n}^{D}$ denote the  restriction of  ${\cal P}_{n}$ to  deterministic strategies ${\cal E}_{n}^{D} \triangleq  \big\{ a_t= e_t(a^{{t-1}},y^{{t-1}}) | t \in  {\mathbb Z}_+^n \big\}$. Define the stochastic optimal control problems with  deterministic strategies,  
\begin{align}
 &J_{n}^{SC}(e^*)\tri  \inf_{ {\cal E}_{n}^{D}} 
 {\bf  E}^e \left\{c_{n}(A^n, X^n)\right\}.\nonumber%\label{threshold}
\end{align}
For strategies ${\cal E}_n^D$ we have $I^P(A^n \rar Y^n)=I^e(A^n \rar Y^n)=0$, because $\frac{ {\bf P}_{Y_t|Y^{t-1}, A^t}^P}{{\bf P}_{Y_t|Y^{t-1}}^P}=0$. Hence,
\begin{align*}
\frac{1}{n}\kappa_{n}(C)  \geq \frac{1}{n}\kappa_{n}(C)\big|_{C=0}=\frac{1}{n}J_{n}^{SC}(e^*) =\kappa_{n,\min}
\end{align*}
where $\frac{1}{n}J_{n}^{SC}(e^*)$ is the classical optimal control of POMDP.

For problem $\frac{1}{n}J_{n}^{SC}(e^*)$, it is easy to verify that $\widehat{\Xi}_t \triangleq \mathbb{P}\{\Xi_t \in d\xi \mid Y^{t-1}\}=\Xi_t, \forall t$, because strategies are deterministic, i.e., $A_t=e_t(A^{t-1}, Y^{t-1})\equiv e_t^{sep}(\Xi_t), \forall t$.
%i.e., 
%with strategies $e \in {\cal E}_n^D$ with zero signalling %$C=0$. 
\end{remark}
%\end{comment}

\section{Conclusion}
We formulated and analyzed the problem of simultaneous signaling and control of POMDPs through the notion of \emph{operational control CC capacity}, $C_{FB}(\kappa)$. We derived meta DP equations whose state space consists of conditional distributions
\[
\widehat{\Xi}_t \triangleq \mathbb{P}\{\Xi_t \in d\xi \mid Y^{t-1}\},
\]
defined over the space of conditional distributions
\[
\Xi_t \triangleq \mathbb{P}\{X_t \in dx \mid A^{t-1}, Y^{t-1}\}.
\]

The randomized strategies admit a separation structure,  expressed in terms of  $\{\Xi_t, \widehat{X}_t\}_{t=1}^n$. Both induced distributions evolve according to Markov recursions and therefore constitute information states. This new reformulation enables the development of numerical algorithms for solving the associated DP equations following similar ideas to \cite{he:2025a,he:2025b} and opens  avenues for computational analysis.

\bibliographystyle{IEEEtran}
\bibliography{string,references_capacity}

\end{document}